\documentclass{eptcs}
\providecommand{\event}{GandALF 2012}
\usepackage{breakurl}
\usepackage[utf8x]{inputenc}
\usepackage[T1]{fontenc}
\RequirePackage{amsfonts,amssymb,amsmath
,ifthen,graphicx
}
\usepackage[all]{xy}
\CompileMatrices

\title{Unambiguous Tree Languages Are Topologically Harder Than Deterministic Ones}
\author{Szczepan Hummel\thanks{This research has been partially supported by the Polish MNiSW grant N N206 567840.}
\institute{Institute of Informatics, University of Warsaw}
}
\def\titlerunning{Unambiguous Tree Languages Are Topologically Harder Than Deterministic Ones}
\def\authorrunning{Szczepan Hummel}
\date{September 2011}

\newcommand{\add}[2]{{\Sigma}^{#1}_{#2}} 	
\newcommand{\mult}[2]{{\Pi}^{#1}_{#2}} 	
\newcommand{\amb}[2]{{\Delta}^{#1}_{#2}} 	

\newcommand{\trees}[1]{T_{#1}}
\DeclareMathOperator{\Borel}{Bor}
\newcommand{\compl}[1]{\overline{#1}}

\newcommand{\aut}[1]{\mathcal{#1}}
\newcommand{\autForLang}[1]{\aut{#1}}	
\newcommand{\autStartsWith}[2]{{#1}_{#2}}	
\newcommand{\trans}[4]{\ensuremath{\xymatrix@=0.2cm{ & {#1} & \\ & *+[o][F-]{#2} \ar@{-}[u] \ar@{-}[dl] \ar@{-}[dr] & \\ {#3} & & {#4} }}}

\newtheorem{theorem}{Theorem}[section]
\newtheorem{lemma}[theorem]{Lemma}
\newtheorem{prop}[theorem]{Proposition}
\newtheorem{remark}[theorem]{Remark}
\newtheorem{corollary}[theorem]{Corollary}

\newenvironment{proof}[1][]{\noindent \textbf{Proof\ifthenelse{\equal{}{#1}}{}{\space}#1:}\par}{\hfill$\square$\bigskip}
\newcommand{\namedpar}[1]{\vspace{0.1cm}\noindent\textbf{#1}}

\newcommand{\Nat}{\omega}
\DeclareMathOperator{\rank}{rank}
\newcommand{\playerA}{\forall}
\newcommand{\playerE}{\exists}
\newcommand{\dual}[1]{\overline{#1}}
\newcommand{\setInShaped}{\ensuremath{W}}

\begin{document}

\maketitle

\begin{abstract}
The paper gives an example of a tree language $G$ that is recognised by an unambiguous parity automaton 
and is $\add 11$-complete (analytic-complete) as a set in Cantor space. 
This already shows that the unambiguous languages are topologically more complex than the deterministic ones,
that are all in $\mult 11$.

Using set $G$ as a building block we construct an unambiguous language that is 
topologically harder than any countable boolean combination of $\add 11$ and $\mult 11$ sets.
In particular the language is harder than any set in difference hierarchy of analytic sets considered by O.~Finkel and P.~Simonnet in the context of nondeterministic automata.


\end{abstract}


\section*{Introduction}
Topological complexity becomes more and more popular as a set complexity measure 
in theoretical computer science, especially in automata theory.
Understanding how hard from the topological point of view are the languages 
recognised by a particular class of automata gives us more understanding 
of the power of those automata. 
It also gives us the access to very powerful and well developed tools coming from descriptive set theory.

One of remarkable uses of descriptive methods is the separation of classes of languages.
Once we know the upper complexity bound for a given class, we can use it while showing that some languages do not belong to the class.
To quote only recent applications of this method, it was used in \cite{max_journal} to show that deterministic max-automata are less expressive than nondeterministic ones.
In \cite{msou_topol_journal}, topological complexity methods were used to exclude a large class of automata as a potential automata model for MSO{+}U logic.

In this paper we address this complexity question for unambiguous 
automata on infinite binary trees, 
i.e. nondeterministic automata that have at most one accepting run on each tree.
With a rise and development of models of automata that do not admit determinisation (infinite tree automata, register automata, BC-automata \cite{bounds}, etc.), the notions like unambiguity, and strong unambiguity, that can be seen as less restrictive variants of determinism, gain the importance. 
For the survey on forms of determinism see \cite{colc_determinism}.

It is well known that deterministic parity tree automata 
recognise only sets in $\mult 11$ topological class (coanalytic sets). On the other hand,
nondeterministic automata recognise some sets that are neither analytic, 
nor coanalytic, but their expressive power is bounded by the second level of 
the projective hierarchy. By Rabin's complementation result (\cite{rabin69}), all 
nondeterministic languages are in $\amb 12$ class, 
i.e. are both $\add 12$ and $\mult 12$.

It was shown by Niwi{\'n}ski and Walukiewicz in \cite{unambiguous_unpublished} (later described in \cite{choice_and_order} and \cite{choice_functions}) that unambiguous automata
do not recognise all nondeterministic languages.
On the other hand, it is not hard to see that they are more expressive than deterministic automata.
An example here might be language $UB$ of trees that have exactly one branch 
with infinitely many labels \emph{a}.

Language $UB$ is a $\mult 11$ complete set. 
However, deterministic automata are also capable of recognising some $\mult 11$ complete sets.
Until now it was not known whether unambiguity introduces any hardness versus determinism from 
the topological viewpoint.
Because $UB$ seems to be typical and close to the very definition of unambiguous automata, it has been widely believed that this example reflects the maximum power of unambiguity.

In this work, in Section \ref{sec:G}, we show an example of unambiguous language $G$ that is $\add 11$ complete, 
hence is not in $\mult 11$. Then, in Section \ref{sec:beyondBC}, using $G$, we construct another unambiguous language
that is topologically harder than any set obtained from $\add 11$ and $\mult 11$ sets by countable boolean operations.
As a consequence, the language itself is not such a boolean combination.
At this level of granularity, this is the most that we can have in locating unambiguous class between deterministic 
and nondeterministic classes. 
To find out if there is a difference in topological complexity between nondeterministic 
and unambiguous languages we would have to refer to more precise complexity measures 
than just the projective hierarchy, e.g. to the finest one - the Wadge hierarchy.

In Section \ref{sec:StronglyUnamb} we show that the example from Section \ref{sec:beyondBC} is actually strongly unambiguous,
i.e. unambiguous together with its complement.

The full version of the paper with Appendix containing a proof of one general topological fact can be found online at:
\url{http://www.mimuw.edu.pl/~shummel/unamb\_topol.pdf}.

\section{Preliminaries}
Let $A$ be an arbitrary set of labels. An infinite binary tree over $A$ is a function $t:\{l,r\}^\ast\to A$.
By $\trees A$ we denote the set of all infinite binary trees over $A$. 
For $t\in\trees A$ and $v\in\{l,r\}^\ast$, we use a notation $t_v$ for the subtree of $t$ rooted in $v$, i.e. $t_v(w)=t(v\cdot w)$.

\subsection{Automata}
A (nondeterministic) \emph{parity tree automaton} over an alphabet $A$ consists of a finite set $Q$ of states, 
a transition relation $\delta\subseteq Q\times A\times Q\times Q$, an initial state $q_0\in Q$, 
and a ranking function $\rank:Q\to\Nat$. We depict a transition $(q,a,q_1,q_2)$ as:
$$\trans{q}{a}{q_1}{q_2}$$
A run $\rho$ of the automaton on a tree $t$ is a labelling of the tree with states ($\rho:\{l,r\}^\ast\to Q$) consistent with the transition relation, 
i.e. for each node $v\in\{l,r\}^\ast$ the tuple $(\rho(v),t(v),\rho(vl),\rho(vr))$ belongs to $\delta$.
We additionally require that the root of $t$ is labelled with an initial state $q_0$.
The run is accepting if on each branch of the tree the highest rank occurring infinitely often in the run is even (the parity condition).

As usual, the language \emph{recognised} by an automaton is the set of all trees on which the automaton has some accepting run.

The automaton is called (top-down) \emph{deterministic} if its transition relation is in fact a function $\delta: Q\times A \to Q\times Q$.
An important property of deterministic automata is that they have exactly one run on each input. 
A similar property gives a rise to consideration of a wider subclass of nondeterministic automata:

The automaton is \emph{unambiguous} if it has at most one accepting run on each input. 
In other words, if it accepts a given tree, it can do it in only one way.

Since in a parity automaton the set of states is finite, so is the set of used ranks (the image of the ranking function).
We say that the automaton is of index $(\iota,\kappa)\in\Nat\times\Nat$ if $\iota\le\rank(Q)\le\kappa$.
Since shifting all ranks by an even number does not change the language recognised by the automaton, it suffices to consider indices $(\iota,\kappa)$ for $\iota\in\{0,1\}$.
We say that the language is of index $(\iota,\kappa)$ if there is an automaton of such index recognising this language.
By identifying each index with the class of languages of this index we obtain the Rabin-Mostowski index hierarchy. 
Lines denote inclusion on the diagram of the hierarchy shown on Figure \ref{fig:index}.

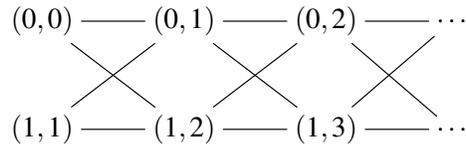
\begin{figure}[htb]
$$\xymatrix{
(0,0) \ar@{-}[r] \ar@{-}[dr] &  (0,1) \ar@{-}[r] \ar@{-}[dr] &  (0,2) \ar@{-}[r] \ar@{-}[dr] & \cdots\\
(1,1) \ar@{-}[r] \ar@{-}[ur] &  (1,2) \ar@{-}[r] \ar@{-}[ur] &  (1,3) \ar@{-}[r] \ar@{-}[ur] & \cdots
}$$
\caption{The Rabin-Mostowski index hierarchy.}\label{fig:index}
\end{figure}

The inclusions come immediately from the definition --- only in some cases we need to shift ranks of corresponding automata by $2$.
Two classes (indices) at the same level of the hierarchy (levels are vertical on the diagram) are called \emph{dual}, and by $\dual{(\iota,\kappa)}$ we denote the class (index) dual to $(\iota,\kappa)$.

Automata of index $(1,2)$ are called \emph{B\"uchi automata}.

If nondeterministic parity automata are considered we talk of \emph{nondeterministic index hierarchy}.
Another often considered and very important variant is \emph{alternating index hierarchy}.
Since we will refer also to this hierarchy, we need to recall shortly the definition of alternating automata.

An \emph{alternating parity tree automaton} is similar to the nondeterministic one with the exception that the set of states is partitioned into two parts $Q=Q_\playerE\cup Q_\playerA$.
The semantics of such an automaton is defined by the game between two players $\playerE$ and $\playerA$.
During a play in this game players construct a run of an automaton in a top-down manner. 
If a given node was labelled by the state from $Q_\playerE$ during this construction then the next transition is chosen by Player $\playerE$, otherwise it is chosen by Player $\playerA$.
The play is won by Player $\playerE$ if in the constructed run on each branch the parity condition holds.
A tree is accepted by the automaton if Player $\playerE$ has a winning strategy in the game defined by this automaton on this tree.

The immediate, but important, fact concerning alternating automata in the context of index hierarchy is:
\begin{remark}\label{rem:alt_dual}
 If a language is of alternating index $(\iota,\kappa)$ then its complement is of alternating index $\dual{(\iota,\kappa)}$.
\end{remark}
\begin{proof}
 It is enough to switch players and shift ranks by one in the automaton recognising a given language.
\end{proof}

The crucial fact about the hierarchies is that they are strict, i.e. all inclusions on Figure \ref{fig:index} are strict. 
This result for nondeterministic hierarchy is due to Niwi{\'n}ski \cite{niwinski86}. 
For alternating hierarchy it was independently proven by Arnold \cite{arnold_altStrict} and Bradfield \cite{bradfield_altStrict}.

\begin{theorem}[Niwi{\'n}ski, Bradfield, Arnold]
 For infinite binary trees, both alternating and nondeterministic index hierarchies are strict. 
\end{theorem}

The languages recognised by nondeterministic (or equivalently alternating) tree automata are called \emph{regular tree languages}.
The languages recognised by deterministic (respectively unambiguous) automata are called deterministic (respectively unambiguous) tree languages.

\subsection{Topology}
For a fixed alphabet $A$, we treat $\trees A$ as a topological space. 
A basic open set is obtained by fixing a finite prefix of trees (a finite tree starting in the root).
Other open sets are obtained by taking arbitrary unions of basic open sets.
If $A$ is finite, this topological space is homeomorphic (i.e. topologically isomorphic) to the Cantor space. 

\subsubsection{Projective hierarchy}
The class $\Borel(X)$ of Borel sets in the topological space $X$ is the least class that:
\begin{itemize}
\item contains all open sets of $X$,
\item is closed under complementation, and
\item is closed under countable unions and intersections.
\end{itemize}
If a space is understood from the context the class is simply denoted by $\Borel$.
  
The class of Borel sets is not closed under projection. 
Each set that is a projection of a Borel set is called \emph{analytic}. 
The class of analytic sets is denoted by $\add 11$. Formally\footnote{The choice of space $T_A$ at the second coordinate is not a commonly made choice in this definition,
but it is best suited for our needs and the resulting notion is the same as in the standard definition.}:
$$\add 11(X) = \left\{P\subseteq X: \exists_{A\textrm{-finite}} \, \exists_{B{\in}\Borel(X{\times}\trees{A})} \; P{=}\pi_1(B)\right\}$$

The rest of the projective hierarchy is defined as follows:
$$\begin{array}{ll}
\mult 1i & \textrm{consists of the complements of the sets from } \add 1i\\
\add 1{i+1} & \textrm{consists of the projections of the sets from } \mult 1i\
\end{array}$$
The sets from the class $\mult 11$ are called \emph{co-analytic}.

An important theorem of Souslin states that if a set is analytic and co-analytic, then it is Borel.

In the sequel we will also use two kinds of intermediate classes. The first kind is:
$$\amb 1i = \add 1i \cap \mult 1i$$
The second one is the $\sigma$-algebra generated by the sets at given level. 
Recall that the $\sigma$-algebra (also called $\sigma$-field) generated by a family $A$ is the closure of the family on countable union, countable intersection and complementation.
 
Figure \ref{fig:projective} present the shape of the projective hierarchy. All the inclusions on the diagram are strict (apart from the leftmost equality coming from Souslin's Theorem).
Proofs of all the facts mentioned in this section can be found in \cite[Chapters 14 and 37]{kechris}.

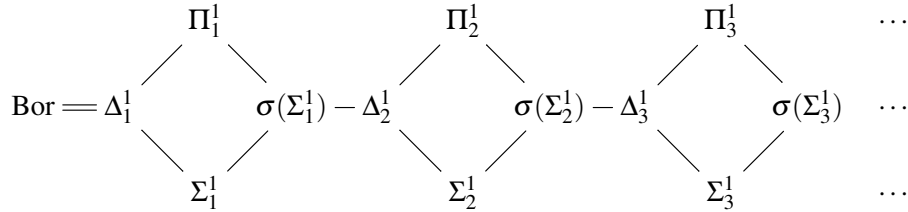
\begin{figure}[htd]
$$\xymatrix@!=.3cm{
& & {\mult 1 1} \ar@{-}[dr] & &  &{\mult 1 2} \ar@{-}[dr] & & & {\mult 1 3} \ar@{-}[dr] & & {\cdots}\\
{\Borel} \ar@{=}[r] & {\amb 11} \ar@{-}[ur] \ar@{-}[dr] & & {\sigma({\add 1 1})} \ar@{-}[r] & {\amb 12} \ar@{-}[ur] \ar@{-}[dr] & & {\sigma({\add 1 2})} \ar@{-}[r] & {\amb 13} \ar@{-}[ur] \ar@{-}[dr] & & {\sigma({\add 1 3})} & {\cdots}\\
& & {\add 1 1} \ar@{-}[ur] & & & {\add 1 2} \ar@{-}[ur] & & & {\add 1 3} \ar@{-}[ur] & & {\cdots}
}$$
\caption{The projective hierarchy}\label{fig:projective}
\end{figure}

\subsubsection{Topological Complexity}
A class of subsets of topological spaces is a \emph{topological complexity class} if it is closed under preimages of continuous functions.
In particular any class depicted on Figure \ref{fig:projective} (if we do not fix any specific space) is a topological complexity class.
Analogously to the complexity theory, there are the notions of \emph{reductions} and \emph{completeness}. 
Let $X$ and $Y$ be two topological spaces and let $K\subseteq X$ and $M\subseteq Y$. A continuous function $f: X \to Y$ 
is a \emph{reduction} of $K$ to $M$ if $K{=}f^{-1}(M)$.
In such case we say that $K$ is \emph{Wadge-reducible} to $M$, or that $M$ is \emph{topologically harder} (\emph{more complex}) than $K$.

For a topological complexity class $\mathbf K$, a set $M$ is called $\mathbf K$-\emph{hard} if any set $K\in \mathbf K$ is Wadge\nobreakdash-re\-du\-cible to $M$.
We say that $M$ is $\mathbf K$\nobreakdash-\emph{complete} if additionally $M\in\mathbf K$.

\subsection{Topological Complexity of Automata}
We recall some basic facts binding automata and index hierarchies with topological hierarchies. 
We use these facts in further discussion.

\begin{theorem} 	
Each regular language of infinite trees is in $\amb 12$.
\end{theorem}
\begin{proof}[(sketch)]
 For a fixed branch $\alpha$ of an infinite tree, the set of runs for which the parity condition on $\alpha$ holds, is Borel.
 Therefore, the set of accepting runs is a $\mult 11$ set --- ``for all branches'' corresponds to co-projection.
 Now, for a fixed nondeterministic parity automaton, the accepted trees are the ones for which there exists an accepting run, so the recognised language is $\add 12$ as a projection of $\mult 11$ set.
 
 By Rabin's complementation lemma (see \cite{rabin69}), the complement of the language recognised by a nondeterministic automaton is also recognised by a nondeterministic automaton. 
 If a language and its complement are both $\add 12$ sets, they are in fact $\amb 12$ sets.
\end{proof}
\begin{theorem} 	
Each language recognised by a deterministic parity tree automaton is in $\mult 11$.
\end{theorem}
\begin{proof}[(sketch)]
 Each deterministic automaton defines a continuous function mapping a tree to the run on it.
 The set of accepted trees is, then, the inverse image of the set of accepting runs under a continuous function. The set of accepting runs is $\mult 11$, then so is the recognised language.
\end{proof}
\begin{theorem}\label{th:compl_buchi}
 Each language recognised by alternating parity automaton of index $(1,2)$ (resp. $(0,1)$) is in $\add 11$ (resp. $\mult 11$) topological class.
\end{theorem}
\begin{proof}
 It was shown by Arnold and Niwi\'nski \cite{FixedPointTrees_AN92} that each language of alternating index $(1,2)$ can be recognised by a nondeterministic automaton of index $(1,2)$ (a B\"uchi automaton).
 Rabin proved in \cite{rabin70} that each such language can be described by existential formula of monadic logic. This implies that they are analytic ($\add 11$).
 
 The fact for index $(0,1)$ comes from the duality --- alternating $(0,1)$ automata recognise complements of sets recognised by alternating $(1,2)$ automata (see Remark \ref{rem:alt_dual}).
\end{proof}

\section{Analytic Complete Language}\label{sec:G}
The result presented in this section 
was inspired by the unpublished work of Bilkowski on the decidability of Unambiguity Problem \cite{bilkowski_personal2010}, 
and by the decidability result presented by Niwi{\'n}ski and Walukiewicz in \cite{gap}.
Bilkowski has shown that the complement of a deterministic language is unambiguous if and only if the language is recognised by a thin automaton, i.e. by an automaton that has only countably many non-trivial paths in each accepting run. 
A path is trivial if, from some moment on, it is labelled only by all-accepting or all-rejecting states
.

The deterministic automaton recognising the complement of language $G$ described below 
is thin in Bilkowski's sense, and has split property --- the sufficient condition for the $\mult 11$ hardness 
from the result of \cite{gap}.
We do not give precise definitions, nor do we discuss the above results since the proofs in this article do not rely on them --- they were only used while constructing the example presented in this section. 

We call a branch of a binary tree over the alphabet $\{a,b\}$ \emph{good} if:
\begin{enumerate}
  \item it is labeled only with \emph{a}'s,\label{cond:a}
  \item it turns left infinitely many times.\label{cond:left}
\end{enumerate}
Let: $$G=\left\{t\in \trees{\{a,b\}}: t \textrm{ has a good branch}\right\}$$

First, we prove the crucial lemma:
\begin{lemma}\label{lemma:leftmost}
 If an infinite binary tree over the alphabet $\{a,b\}$ 
has a good branch, then it has the left-most such branch, 
i.e. a good branch such that there is no good branch 
to the left.
\end{lemma}
\begin{proof}
 Assume that a tree $t$ has a good branch.
The construction of the left-most good branch goes as follows.
We start from the root.
If we have constructed the prefix of the branch up to the node $v$ we advance 
to the left descendant if there are good branches 
going through it. Otherwise we advance to the right descendant.
Call the branch constructed by this procedure $\rho(t)$.

By the construction, it is clear that there is no good branch to the left from $\rho(t)$.
Now we prove that $\rho(t)$ is good.
Note that during the construction we maintain the invariant that there is a good branch 
 going through a considered node. 
In particular, all nodes we have selected are labeled with \emph{a},
therefore, we only need to verify property \ref{cond:left} from the definition to proof goodness 
of $\rho(t)$. 

Assume that $\rho(t)$ turns left only finitely many times.
Then there is a vertex $v$ on the branch $\rho(t)$ after which $\rho(t)$ turns only right.
Let us take a good branch going through $v$, and call it $\sigma$. 
By the assumption, $\rho(t)$ is not good, so branches $\rho(t)$ and $\sigma$ diverge in some vertex $w$.
Since $w$ is below $v$, $\rho(t)$ goes right from $w$ and $\sigma$ goes left.
Since $\sigma$ is good, the construction should have selected the left
descendant of $w$, but have selected the right one. That yields a contradiction,
so $\rho(t)$ turns left infinitely many times.
\end{proof}

Now let $L=\compl{G}$.

\begin{prop}\label{prop:G-unamb}
 Language $G$ is recognised by an unambiguous automaton and its complement $L$ is recognised by a deterministic automaton.
\end{prop}
\begin{proof}
 Thanks to Lemma \ref{lemma:leftmost}, to prove unambiguity of $G$ it is enough to show 
an automaton that guesses the left-most good branch and verifies 
correctness of the guess.
The idea is that automaton goes along a branch labelled with $a$'s, 
proving that the branch turns left infinitely many times,
and proving that everything that diverges to the left from the branch does not have a good branch (i.e. belongs to $L$),
and not caring what happens to the right from the branch.

Let us then start with constructing a deterministic automaton $\autForLang L$ recognising $L$.
It has 3 states: states $l_0$ and $l_1$ (subscript indicates rank) occur on paths that have had only $a$'s so far, 
and track turns to the left; state $\top_0$ is all-accepting (i.e. self-looping with rank $0$).
Initial state is $l_1$ and the transitions are as follows:
$$
\trans{l_\ast}{a}{l_1}{l_0} \;
\trans{l_\ast}{b}{\top_0}{\top_0}
$$
The automaton uses ranks $\{0,1\}$. 
Note that, for given tree $t$, the run of $\autForLang L$ on $t$ has a branch with infinitely many ranks $1$ on it if and only if there is an $a$-labelled branch turning left infinitely often in $t$. Therefore the automaton recognises $L$.

The automaton $\autForLang G$ for language $G$ uses $\autForLang L$ as a component. It has 7 states and uses ranks $0,1,2$.
States $g_1$ and $g_2$ (again, subscript indicates rank) are used to track the branch;
states of $\autForLang L$ --- to prove non-existence of a good branch in a subtree;
state $\top_2$ is all-accepting, and state $\bot_1$ is all-rejecting.
The initial state is $g_1$. The automaton uses the following transitions:
$$
\begin{array}{cc}
 & L\textrm{-part} \\
\cline{2-2}
\trans{g_\ast}{a}{g_2}{\top_2} \;
\trans{g_\ast}{a}{l_1}{g_1} \;
\trans{g_\ast}{b}{\bot_1}{\bot_1} \quad
&
  \multicolumn{1}{|c|}{
    \trans{l_\ast}{a}{l_1}{l_0} \;
    \trans{l_\ast}{b}{\top_0}{\top_0}
  } \\
\cline{2-2}
\end{array}
$$

It is not hard to see that presented automaton implements described idea, 
therefore accepts if and only if given tree has a good branch.
It is unambiguous, because it only can accept by labelling the left-most good branch with $g$ states.
\end{proof}
\begin{remark}\label{rem:indOfGL}
 Automaton $\autForLang{L}$ is of index $(0,1)$, and $\autForLang{G}$ is of index $(0,2)$.

 Note that automaton $\autForLang{G}$ can be transformed into an equivalent one of index $(1,3)$, by making $L$\nobreakdash-part use ranks $\{2,3\}$ instead of $\{0,1\}$.
\end{remark}

\begin{prop}\label{prop:G-compl}
 Set $G$ is $\add 11$ complete.
\end{prop}
\begin{proof}
 To prove the hardness we continuously reduce the set $IF$ 
of $\omega$-branching trees with an infinite branch to our set $G$.
Set $IF$ is a well known $\add 11$-complete subset of the space $Tr$ of trees on $\omega$, 
i.e. prefix-closed subsets of $\omega^\ast$ (see e.g. \cite[Theorem 27.1]{kechris}).
The topology on $Tr$ is similar to the one on $\trees A$: a basic open set is obtained by fixing some finite part of trees.
E.g. for some integer $n$, we fix what nodes out of $\{1,...,n\}^{\le n}$ belong and what do not belong to all trees in a set.

We construct a reducing function $f:Tr\to \trees{\{a,b\}}$.
Fix a tree $t\in Tr$. Put labels $a$ to the root and all right descendant nodes in the tree $f(t)$.
For each node $n_1n_2n_3\dots n_k$ of tree $t$ we put label $a$ to the node 
$r^{n_1}lr^{n_2}lr^{n_3}l\dots r^{n_k}l$ in $f(t)$. Remaining left descendant nodes obtain label $b$.

Note that $f(t)$ has an $a$-labelled branch that turns left infinitely many times (i.e. a good branch)
if and only if $t$ has an infinite branch.
Therefore
$$f(t)\in G \iff t\in IF$$
So $f$ indeed reduces $IF$ to $G$. 

Function $f$ is continuous, because the labels at $n$'th level of the tree $f(t)$ are determined by the finite part of a tree $t$, namely the part in $\{1,...,n\}^{\le n}$.

The upper topological complexity bound of set $G$ comes from Proposition \ref{prop:G-Buchi}, that will be proven later, and from Theorem \ref{th:compl_buchi}.
\end{proof}

Therefore, we have proven the following:
\begin{theorem}
 There is a $\add 11$ complete language of infinite trees, that is recognised by an unambiguous parity automaton.
\end{theorem}


Recall the theorem that is stated as Corollary 4.14 in \cite{finkel_nonborel_reg}:
\begin{theorem}[Finkel, Simonnet]\label{thm:finkel_unambBuchi}
 A tree language recognised by an unambiguous B\"uchi automaton is Borel.
\end{theorem}

Thanks to this theorem we know that no unambiguous B\"uchi automaton recognises language $G$.
Proposition \ref{prop:G-compl} implies that $G$ is not a $\mult 11$ set, therefore, by Theorem \ref{th:compl_buchi}, it cannot be recognised by any (even alternating) automaton of index $(0,1)$.
As a result we obtain that the use of $3$ priorities is necessary for unambiguous automaton to recognise $G$.
On the other hand observe that:

\begin{prop}\label{prop:G-Buchi}
 Language $G$ is recognised by a nondeterministic B\"uchi automaton.
\end{prop}
\begin{proof}
 It suffices to remove $L$-part from the unambiguous automaton $\autForLang G$ presented in the proof of Proposition \ref{prop:G-unamb}, replacing $l_1$ with $\top_2$ in other transitions to obtain needed nondeterministic automaton. 
 The only purpose of that part was to make sure that we select left-most good branch. We do not need this if we do not care about the number of accepting runs.
\end{proof}

From this observation we obtain that the result by Finkel and Simonnet is tight in the sense that the following strengthening of Theorem \ref{thm:finkel_unambBuchi} does NOT hold:
``An unambiguous language that is recognised by some (possibly ambiguous) nondeterministic B\"uchi automaton, is Borel''.
Let us state it as follows:

\begin{corollary}
 There is a language of non-Borel topological complexity that is on one hand unambiguous, and on the other hand B\"uchi.
\end{corollary}

\section{Beyond Boolean Combinations}\label{sec:beyondBC}
In this section we construct an unambiguous tree language that is topologically harder than any set in 
$\sigma(\add 11)$.
For that we need to prove that the class of the sets that reduce to the constructed language contains all analytic sets, and is closed under complementation and countable unions. 


The construction of the language harder than boolean combinations of analytic sets goes through an automaton. The automaton, we will call it $\aut C$, uses as building blocks:
\begin{enumerate}
 \item the unambiguous automaton $\autForLang G$ recognising language $G$ from Section \ref{sec:G},
 \item the deterministic automaton $\autForLang L$ recognising $L:=\compl{G}$.
\end{enumerate}
Since we will work with larger alphabet than just $\{a,b\}$, we modify the automata (and recognised languages), such that they treat all letters except of $a$ like $b$.
After this modification, the languages recognised by the automata are complements even in case of larger alphabet.

The idea is that automaton $\aut C$ expects the part of a tree near the root to be shaped as a formula defining some set in the $\sigma$\nobreakdash-algebra $\sigma(\add 11)$.
Countable union is represented by the branch turning only left, where right descending subtrees correspond to subformulas.
Complementation nodes in a way disregard the left descendants.
In nodes corresponding to the atoms of the formula (analytic sets) the automaton expects subtrees from language $G$. The details follow.

The automaton $\aut C$ works over the alphabet $A=\{a,b,\vee,\neg\}$.
Apart from the states of the automata $\autForLang G$ and $\autForLang L$, it uses states: $N$, $P$, $N_\vee$, $P_\vee$, $\top_\vee$, $\top_S$, $\top$ ($P$ stands for 'Positive', $N$ for 'Negative'; $\top_\vee$ and $\top_S$ serve only verifying the shape of the formula). The initial state is $P$. The transitions of $\aut C$ are as follows: 
$$\begin{array}{ccccc}
\trans{P,P_\vee}{\vee}{\top_\vee}{P} &
\trans{P,P_\vee}{\vee}{P_\vee}{N} &
\trans{N,N_\vee}{\vee}{N_\vee}{N} &
\trans{\top_\vee,\top_S}{\vee}{\top_\vee}{\top_S} \\
\trans{N}{\neg}{\top}{P} &
\trans{P}{\neg}{\top}{N} &
\trans{\top_S}{\neg}{\top}{\top_S} &
\trans{\top_S}{a,b}{\top}{\top} 
\end{array}
$$
where $\top$ is an all-accepting state.
Additionally, when the automaton encounters letter $a$ or $b$ in state $P$, it starts to act as automaton $\autForLang G$. If it encounters one of these letters in state $N$, it starts to act as automaton $\autForLang L$.

States $N$, $P$, $P_\vee$, $\top_S$ have rank $1$, states $N_\vee$, $\top_\vee$, $\top$ rank $0$.

Note that automaton $\aut C$ is of index $(0,2)$ --- it is because of ranks used by $\autForLang G$.

We will use the notation $\autStartsWith{\aut{A}}{q}$ for the automaton $\aut A$ modified in such a way, that $q$ becomes an initial state. Let us prove the following:

\begin{prop}\label{prop:nonBC}
 For each set $S\subseteq X$ obtained from analytic sets by countable boolean operations the following holds:
 $$\left.\parbox[l]{0.8\linewidth}{
  there is a continuous function $f:X\to \trees A$, that simultaneously reduces $S$ to $L(\aut C)$, and $\compl S$ to  $L(\autStartsWith{\aut C}{N})$, i.e.:
  \begin{enumerate}
    \item $\aut C$ accepts $f(x)$ iff $x\in S$,
    \item $\autStartsWith{\aut C}{N}$ accepts $f(x)$ iff $x\notin S$.

Additionally we require:
    \item\label{cond:T_S} For all $x$, $\autStartsWith{\aut C}{\top_S}$ accepts $f(x)$.
  \end{enumerate}
  
 }\right\} (*)$$
\end{prop}
\begin{proof}
 We need to prove that the class of sets for which the property $(*)$ holds contains all analytic sets, and is closed under complementation and countable unions.

 \namedpar{(analytic sets)}
 Take an arbitrary analytic set $S\subseteq X$. Since $G$ is analytic-complete, there exists a continuous function $f:X\to \trees{\{a,b\}}$ reducing $S$ to $G$. 
 We will show that the same function is actually a needed simultaneous reduction to $L(\aut{C})$ and $L(\autStartsWith{\aut C}{N})$.
 
 Note that among the trees over the alphabet $\{a,b\}$, $\aut{C}$ accepts exactly the trees from $G$. Therefore:
 $$f(x)\in L(\aut{C}) \quad\iff\quad f(x)\in G \quad\iff\quad x\in S$$
 Now we have to show that $\aut{C}$ accepts $f(x)$ from state $N$ if and only if $x\notin S$.
 Recall that from state $N$, on trees over $\{a,b\}$, $\aut{C}$ acts exactly like automaton $\autForLang L$. 
 Since $\autForLang L$ recognises the complement of $G$, we obtain the statement.
 
 From state $\top_S$ each tree with $a$ or $b$ in the root is accepted, so condition \ref{cond:T_S} also holds for $f$.

 \namedpar{(closure under the complement)}
 Take any set $S\subseteq X$ for which $(*)$ holds, and let $f:X\to \trees A$ be an appropriate reducing function.
 We construct a reduction $\compl{f}:X\to \trees A$ of $\compl S$ to $L(\aut{C})$ by putting:

$\compl{f}(x) := $ \parbox{0.3\textwidth}{\includegraphics[width=0.2\textwidth]{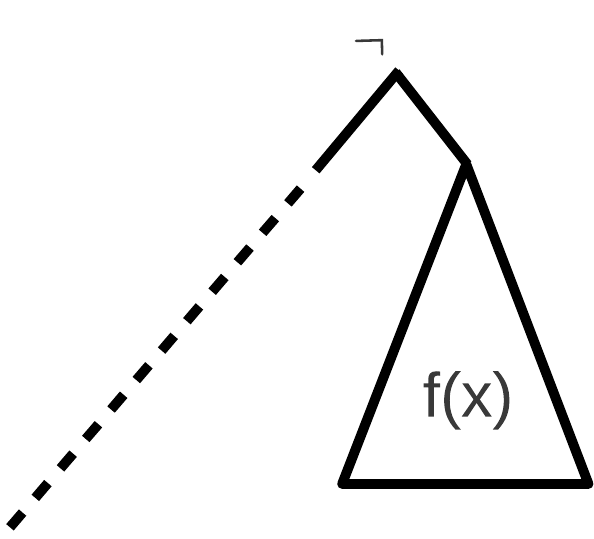}} \hspace{.5cm} \parbox{0.5\textwidth}{(the rest of the leftmost branch does not matter)}

 By the inductive assumption and by the shape of the transitions over letter $\neg$, we obtain that $\compl{f}(x)$ is accepted from state $P$ if and only if $f(x)$ is accepted from state $N$, that $\compl{f}(x)$ is accepted from state $N$ if and only if $f(x)$ is accepted from state $P$, and that $\compl{f}(x)$ is accepted from state $\top_S$ if and only if $f(x)$ is accepted from state $\top_S$.

 \namedpar{(closure under unions)}
 Take a set $S\subseteq X$ such that $S=\bigcup_{i=0}^\infty S_i$, where each $S_i$ has property $(*)$.
 For each $i$, let $f_i:X\to \trees A$ be an appropriate reduction as in $(*)$ for set $S_i$.
 We build a reduction for set $S$ as follows:
 
\begin{center}
 $f(x) := $ \parbox{0.3\textwidth}{\includegraphics[width=0.35\textwidth]{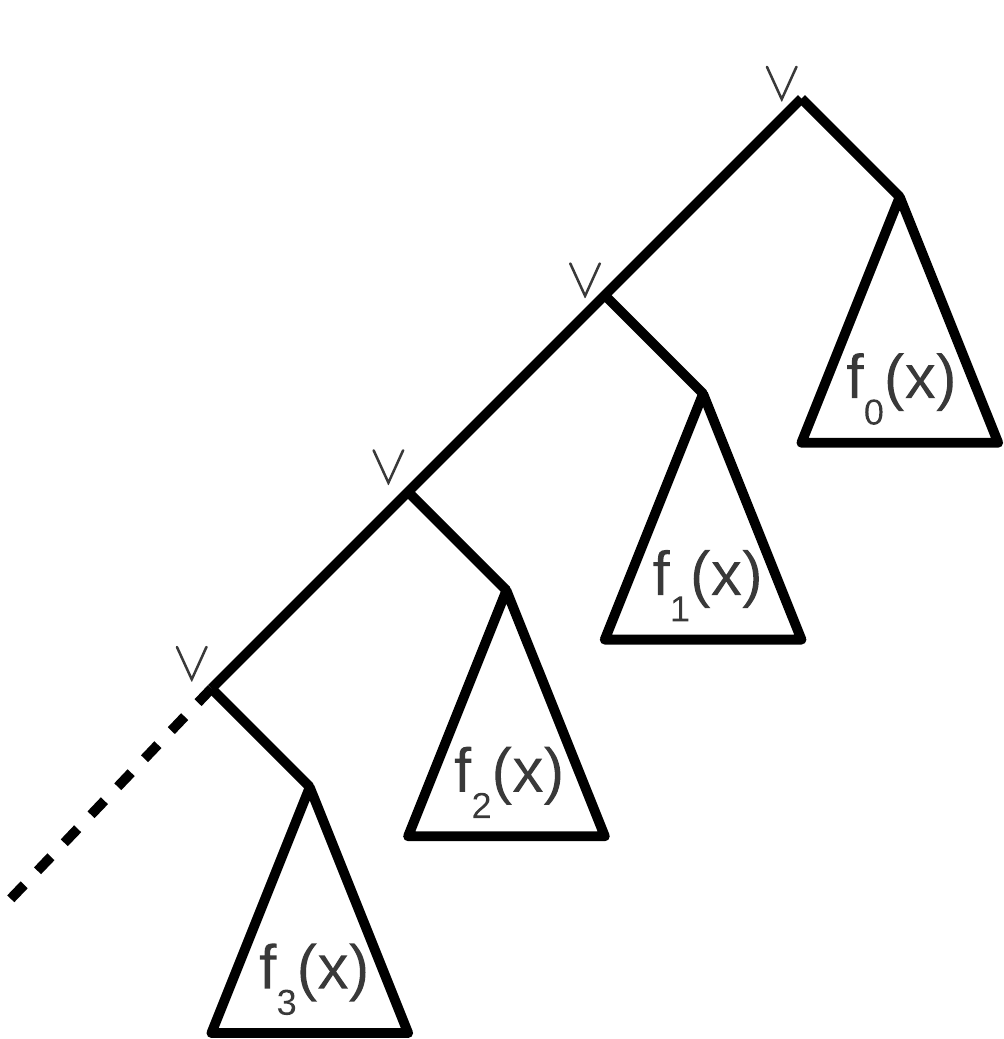}} 
\end{center}
 
 Function $f$ is clearly continuous. We now prove that it is a needed simultaneous reduction.
 
 Let $x\in S$, and let $i_0$ be the least $i$ such that $x\in S_i$.
 Then, from the inductive assumption: $f_{i_0}(x)$ is accepted from state $P$; for each $i<i_0$, $f_i(x)$ is accepted from state $N$; for each $i>i_0$, $f_i(x)$ is accepted from state $\top_S$.
 Then, by the shape of the transitions over label $\vee$, $f(x)$ is accepted from state $P$.

 Let now $x\notin S$. Then $x\notin S_i$ for each $i$.
 Therefore, for each $i$, $f_i(x)$ is accepted from state $N$.
 Then $f(x)$ is accepted from state $N$, by the shape of transitions over $\vee$.
 
 If each $f_i(x)$ is accepted from state $\top_S$, then by the shape of the transitions from $\top_S$ and $\top_\vee$ over $\vee$, $f(x)$ is also accepted from state $\top_S$.
\end{proof}

In particular we have proven:
\begin{corollary}
 Each set from $\sigma(\add 11)$ continuously reduces to $L(\aut C)$.
\end{corollary}

We use one more fact from descriptive set theory. It can be proven using Theorems 1E.3, 1D.2, 1D.3 from \cite{moschovakis} and the fact that $\sigma(\add 11)$ is closed under the complement 
(the details can be found in Appendix of the full version of this paper).

\begin{prop}\label{prop:nocomplete}
 For each $n\ge 1$, there is no $\sigma(\add 1n)$\nobreakdash-complete set.
\end{prop}

\begin{corollary}
 Set $L(\aut C)$ is in $\amb 12 \setminus \sigma(\add 11)$.
\end{corollary}

Let now $\setInShaped\subseteq\trees{A}$ be the language of trees corresponding to improperly shaped formulas,
i.e. the set of trees that:
\begin{enumerate}
 \item\label{def:W_notWF} contain a branch labelled only with $\vee$ and $\neg$, that turns right infinitely many times, and turns right after each occurrence of $\neg$ (formula is not well-founded),
 \item[or]
 \item\label{def:W_veeStoped} contain a path from the root that turns right after each occurrence of $\neg$, ends in a left descendant node $v$ that is labelled with something different than $\vee$,
a parent of $v$ is labelled with $\vee$, and all nodes above $v$ are labelled with $\vee$ or $\neg$.
\end{enumerate}

\begin{lemma}\label{l:disjoint}
 Sets $L(\aut C)$, $L(\autStartsWith{\aut{C}}{N})$ and $\setInShaped$ constitute a partition of $\trees{A}$.
 In particular $L(\aut C)$ and $L(\autStartsWith{\aut{C}}{N})$ are disjoint.
\end{lemma}
\begin{proof}
 First we show that if $t\in\setInShaped$, then $t$ is not accepted from neither of states $P$, $N$. 
Assume first, that $t$ has a branch $\rho$ as in point \ref{def:W_notWF} of the definition of $\setInShaped$. Consider a run of the automaton starting from $P$ or $N$. Branch $\rho$ is labelled with states $N$, $P$, $N_\vee$, $P_\vee$, $\top_\vee$, $\top_S$ (ranks $0$ and $1$) in this run, and after each turn to the right there is state $N$, $P$, or $\top_S$ (each of rank $1$). The highest rank occurring infinitely often on branch $\rho$ is $1$, so the run is not accepting.
Now assume that $t$ has a path as in point \ref{def:W_veeStoped} of the definition of $\setInShaped$. In a run starting from $P$ or $N$, the node $v$, as a left descendant of a node labelled with $\vee$, always obtains a state label of $N_\vee$, $P_\vee$, or $\top_\vee$. Since from these states there are only transitions over letter $\vee$ and $v$ is labelled with something different, the run gets stuck.
 
Let us now consider well-shaped trees, i.e. trees outside $\setInShaped$. Recall that prefixes of such trees represent well-founded formulas.
Using this fact we define a rank on such trees. For a tree $t$, $\rank(t)$ is an ordinal defined by the rules:
\begin{itemize}
 \item If the root of $t$ is labelled with $a$ or $b$, then $\rank(t)=0$.
 \item If the root is labelled with $\neg$, then $\rank(t)=\rank(t_r)+1$, where $t_r$ is a subtree of $t$ rooted in right descendant node of the root.
 \item If the whole left-most branch is labelled with $\vee$, then $\rank(t)=sup\left\{\rank(t_{l^nr})+1:n\ge 0\right\}$, where $t_{l^nr}$ are subtrees diverging from the path to the right.
\end{itemize}
Note that the above set of rules allows to define the rank for all well-shaped trees.

Now we prove by the transfinite induction on $\rank(t)$ that each well-shaped tree $t$ is either accepted from state $P$ or from $N$, but not from both.

If $\rank(t)=0$, the tree has $a$ or $b$ in the root. It is, then, accepted from state $P$ if and only if $t\in G$, and is accepted from state $N$ if and only if $t\in L$.
Since $L$ is the complement of $G$, we are done with this case.

If $\rank(t)>0$, then the root is labelled with $\vee$ or $\neg$. If it is labelled with $\neg$, then $t$ is accepted from state $P$ if and only if $t_r$ is accepted from state $N$, and vice versa.
Since $t_r$ is of lower rank than $t$, it suffices to use the inductive assumption for $t_r$.

If $t$ has label $\vee$ in the root, then, since it is well-shaped, it has label $\vee$ at whole left-most branch. All subtrees diverging to the right from the left-most branch have lower rank then $t$, so we can use the inductive assumption for them. Then: 
\begin{enumerate}
 \item[either] at least one of the subtrees is accepted from state $P$ and not from $N$ --- in this case $t$ is accepted from $P$, but not from $N$,
 \item[or] each of them is accepted from state $N$ and not from $P$ --- then tree $t$ is accepted from state $N$, but not from $P$.
\end{enumerate}
\end{proof}
\begin{prop}
 Automaton $\aut C$ is unambiguous.
\end{prop}
\begin{proof}
 The only nondeterminism in automaton $\aut C$ occurs in state $P$ (or $P_\vee$) and label $\vee$.
 Two transitions are possible from this configuration --- one assigns state $P$ to the right child, the other one assigns state $N$ to the right child.
 By Lemma \ref{l:disjoint}, only one of the transitions can be used in an accepting run from a given node in a given tree.
\end{proof}

Recall that automaton $\aut C$ is of index $(0,2)$. 
Since $L(\aut C)$ is neither a $\add 11$ set, nor a $\mult 11$ set, by Theorem \ref{th:compl_buchi}, we have:
\begin{remark}
 Language $L(\aut C)$ cannot be recognised by any alternating automaton using $2$ ranks (neither $(0,1)$, nor $(1,2)$).
\end{remark}

\section{The Example Is Strongly Unambiguous}\label{sec:StronglyUnamb}
A language $M$ is called \emph{strongly unambiguous} if both $M$ and $\compl{M}$ are recognised by unambiguous automata. 
Because of the apparent asymmetry of the class of unambiguous languages, strong unambiguity seems to be an important notion.
Some arguments for its importance were given by Colcombet in \cite{colc_determinism}.

In this section we prove that the example from Section \ref{sec:beyondBC} is actually strongly unambiguous.

\begin{prop}
 The complement of language $L(\aut C)$ is recognised by an unambiguous tree automaton.
\end{prop}
\begin{proof}
 Recall that, by Lemma \ref{l:disjoint}, $\compl{L(\aut C)}$ is the disjoint union of $\setInShaped$ and $L(\autStartsWith{\aut C}{N})$.
 Since the class of unambiguous languages is closed under disjoint unions, and $\autStartsWith{\aut C}{N}$ is an unambiguous automaton, we only need to show how to recognise incorrectly shaped trees unambiguously. It can be done by finding the rightmost incorrect branch. The proof that the rightmost incorrect branch exists is analogous to the proof of Lemma \ref{lemma:leftmost}.

The automaton uses states $B$, $B_\vee$, $\top_S$, $\top_\vee$, $\top$. States $B$ and $B_\vee$ are used to track the branch and $\top_S$ serves showing non-existence of incorrect branch in a subtree.
State $B$ is of rank $2$; $\top_S$ and $B_\vee$ of rank $1$; the remaining states are of rank $0$. The initial state is $B$. The transitions are as follows:
$$
\begin{array}{cc}
 & \textrm{as in }\aut C \\
\cline{2-2}
\trans{B,B_\vee}{\vee}{\top}{B} \;
\trans{B,B_\vee}{\vee}{B_\vee}{\top_S} \;
&
  \multicolumn{1}{|c|}{
    \trans{\top_\vee,\top_S}{\vee}{\top_\vee}{\top_S} \;
    \trans{\top_S}{\neg}{\top}{\top_S} \;
  } \\
\trans{B}{\neg}{\top}{B} \;
\trans{B_\vee}{-\vee}{\top}{\top}
&
  \multicolumn{1}{|c|}{
    \trans{\top_S}{a,b}{\top}{\top}
  }\\
\cline{2-2}
\end{array}
$$
where $-\vee$ is any label different than $\vee$.

Each run of the automaton has exactly one branch that is labelled (whole or up to some node) with states $B$ and $B_\vee$.
The ranks of the branch fulfill the parity condition if a node with label different than $\vee$ occurs as a left descendant of a $\vee$-node
--- then the branch is as in point $2$ of the definition of set $\setInShaped$; 
or if it turns right infinitely many times and is whole labelled with $\vee$ or $\neg$ --- then it reflects point $1$ of the definition of set $\setInShaped$.
The subtrees diverging to the right from the branch are well shaped --- as in automaton $\aut C$, $\top_S$ guarantees this. 
Therefore, the automaton is unambiguous --- the right-most incorrect branch has to be selected.
\end{proof}

The following theorem summarises the topological results of this paper:
\begin{theorem}
 There is a strongly unambiguous tree language that is not in $\sigma(\add 11)$ class.
\end{theorem}

\section{Related Work}\label{sec:related}
In \cite{finkel_nonborel_reg} the authors consider the difference hierarchy of analytic sets.
They show a sequence of regular tree languages hard for the levels $D_\alpha(\add 11)$ of the hierarchy, for $\alpha<\omega^\omega$.
Then, on page 10, they note that the languages (even the first one in the sequence) are not recognised by unambiguous tree automata.

The base for the construction of Finkel and Simonnet is the set:
$$\{ t\in \trees{\{0,1\}} : \textrm{there is a branch in $t$ with infinitely many labels 1} \}$$

If we look at the construction there, we see that we can use unambiguous set $G$ from this paper as a base, and hardness results still hold.
Actually the proofs remain exactly the same, since they only use $\add 11$-hardness of the basic set.

Now we note that each automaton built during the construction is unambiguous if we use unambiguous automaton $\autForLang{G}$ recognising basic set $G$ and unambiguous automaton $\autForLang{L}$ recognising its complement (i.e. set $L$ from Section \ref{sec:G}) as basic building blocks. 
Indeed, for given tree $t$, the automaton described in the proof of Lemma 4.5 in \cite{finkel_nonborel_reg} always selects the path corresponding to the smallest ordinal $\omega^{n-1}\cdot a_{n-1} + \omega^{n-2}\cdot a_{n-2} + \dots + \omega\cdot a_1 + a_0$ for which the tree $t_{l^{a_{n-1}}rl^{a_{n-2}}r\cdots rl^{a_0}}$ belongs to $G$, and proves that all paths corresponding to less ordinals end up in something that does not belong to $G$ (i.e. belongs to $L$). Unambiguous automaton $\autForLang{G}$ is used to verify belonging to $G$ and $\autForLang{L}$ is used to verify belonging to $L$.

As a result, using example languages $G$ and $L$ from this paper and the construction from the paper \cite{finkel_nonborel_reg} by Finkel and Simonnet, we get a sequence of unambiguous languages hard for the classes $D_\alpha(\add 11)$, for $\alpha<\omega^\omega$. 
We give this fact here only as a note, because all the languages in the sequence reduce to the language from Section \ref{sec:beyondBC}.
This is because each set at any countable level of the difference hierarchy of analytic sets is, by the definition, in $\sigma(\add 11)$.

We add a note by an anonymous reviewer. The game tree language $W_{(0, 2)}$ (one of the languages used in \cite{arnold_altStrict} to show strictness of alternating index hierarchy) does not belong to the class $\sigma(\Sigma_1^1)$, which is an answer to the question asked in the last paragraph of \cite{finkel_nonborel_reg}. Indeed, language $L(C)$ is of index $(0, 2)$ and thus is Wadge\nobreakdash-reducible to language $W_{(0, 2)}$ by the result cited in \cite[Lemma 5.2]{finkel_nonborel_reg}.

\section*{Acknowledgements}
The author wants to thank Micha{\l} Skrzypczak and Henryk Michalewski for their support and advice in some topological proofs.
Many thanks also go to Damian Niwi\'nski and the anonymous reviewers for very valuable comments that have led to a significant improvement of the paper.

\bibliographystyle{eptcs}
\bibliography{bib}

\begin{thebibliography}{10}
\providecommand{\bibitemdeclare}[2]{}
\providecommand{\surnamestart}{}
\providecommand{\surnameend}{}
\providecommand{\urlprefix}{Available at }
\providecommand{\url}[1]{\texttt{#1}}
\providecommand{\href}[2]{\texttt{#2}}
\providecommand{\urlalt}[2]{\href{#1}{#2}}
\providecommand{\doi}[1]{doi:\urlalt{http://dx.doi.org/#1}{#1}}
\providecommand{\bibinfo}[2]{#2}

\bibitemdeclare{article}{arnold_altStrict}
\bibitem{arnold_altStrict}
\bibinfo{author}{Andr{\'e} \surnamestart Arnold\surnameend}
  (\bibinfo{year}{1999}): \emph{\bibinfo{title}{The $\mu$-calculus
  alternation-depth hierarchy is strict on binary trees}}.
\newblock {\sl \bibinfo{journal}{ITA}}
  \bibinfo{volume}{33}(\bibinfo{number}{4/5}), pp. \bibinfo{pages}{329--340}.
\newblock \urlprefix\url{http://dx.doi.org/10.1051/ita:1999121}.

\bibitemdeclare{incollection}{FixedPointTrees_AN92}
\bibitem{FixedPointTrees_AN92}
\bibinfo{author}{Andr{\'e} \surnamestart Arnold\surnameend} \&
  \bibinfo{author}{Damian \surnamestart Niwi\'nski\surnameend}
  (\bibinfo{year}{1992}): \emph{\bibinfo{title}{Fixed point characterization of
  weak monadic logic definable sets of trees}}.
\newblock In: {\sl \bibinfo{booktitle}{Tree Automata and Languages}}, pp.
  \bibinfo{pages}{159--188}.

\bibitemdeclare{misc}{bilkowski_personal2010}
\bibitem{bilkowski_personal2010}
\bibinfo{author}{Marcin \surnamestart Bilkowski\surnameend}
  (\bibinfo{year}{2010}): \bibinfo{howpublished}{personal communication}.

\bibitemdeclare{article}{max_journal}
\bibitem{max_journal}
\bibinfo{author}{Miko{\l}aj \surnamestart Boja\'nczyk\surnameend}
  (\bibinfo{year}{2011}): \emph{\bibinfo{title}{Weak {MSO} with the Unbounding
  Quantifier}}.
\newblock {\sl \bibinfo{journal}{Theory Comput. Syst.}}
  \bibinfo{volume}{48}(\bibinfo{number}{3}), pp. \bibinfo{pages}{554--576}.
\newblock \urlprefix\url{http://dx.doi.org/10.1007/s00224-010-9279-2}.

\bibitemdeclare{inproceedings}{bounds}
\bibitem{bounds}
\bibinfo{author}{Mikolaj \surnamestart Boja\'nczyk\surnameend} \&
  \bibinfo{author}{Thomas \surnamestart Colcombet\surnameend}
  (\bibinfo{year}{2006}): \emph{\bibinfo{title}{Bounds in
  $\omega$-Regularity}}.
\newblock In: {\sl \bibinfo{booktitle}{LICS}}, pp. \bibinfo{pages}{285--296}.
\newblock
  \urlprefix\url{http://doi.ieeecomputersociety.org/10.1109/LICS.2006.17}.

\bibitemdeclare{article}{bradfield_altStrict}
\bibitem{bradfield_altStrict}
\bibinfo{author}{Julian~C. \surnamestart Bradfield\surnameend}
  (\bibinfo{year}{1999}): \emph{\bibinfo{title}{Fixpoint alternation:
  {A}rithmetic, transition systems, and the binary tree}}.
\newblock {\sl \bibinfo{journal}{ITA}}
  \bibinfo{volume}{33}(\bibinfo{number}{4/5}), pp. \bibinfo{pages}{341--356}.
\newblock \urlprefix\url{http://dx.doi.org/10.1051/ita:1999122}.

\bibitemdeclare{inproceedings}{choice_and_order}
\bibitem{choice_and_order}
\bibinfo{author}{Arnaud \surnamestart Carayol\surnameend} \&
  \bibinfo{author}{Christof \surnamestart L{\"o}ding\surnameend}
  (\bibinfo{year}{2007}): \emph{\bibinfo{title}{{MSO} on the Infinite Binary
  Tree: Choice and Order}}.
\newblock In: {\sl \bibinfo{booktitle}{CSL}}, pp. \bibinfo{pages}{161--176}.
\newblock \urlprefix\url{http://dx.doi.org/10.1007/978-3-540-74915-8_15}.

\bibitemdeclare{article}{choice_functions}
\bibitem{choice_functions}
\bibinfo{author}{Arnaud \surnamestart Carayol\surnameend},
  \bibinfo{author}{Christof \surnamestart L{\"o}ding\surnameend},
  \bibinfo{author}{Damian \surnamestart Niwi\'nski\surnameend} \&
  \bibinfo{author}{Igor \surnamestart Walukiewicz\surnameend}
  (\bibinfo{year}{2010}): \emph{\bibinfo{title}{Choice functions and
  well-orderings over the infinite binary tree}}.
\newblock {\sl \bibinfo{journal}{Central European Journal of Mathematics}}
  \bibinfo{volume}{8}, pp. \bibinfo{pages}{662--682}.
\newblock \urlprefix\url{http://dx.doi.org/10.2478/s11533-010-0046-z}.

\bibitemdeclare{inproceedings}{colc_determinism}
\bibitem{colc_determinism}
\bibinfo{author}{Thomas \surnamestart Colcombet\surnameend}
  (\bibinfo{year}{2012}): \emph{\bibinfo{title}{Forms of Determinism for
  Automata (Invited Talk)}}.
\newblock In: {\sl \bibinfo{booktitle}{STACS}}, pp. \bibinfo{pages}{1--23}.
\newblock \urlprefix\url{http://dx.doi.org/10.4230/LIPIcs.STACS.2012.1}.

\bibitemdeclare{article}{finkel_nonborel_reg}
\bibitem{finkel_nonborel_reg}
\bibinfo{author}{Olivier \surnamestart Finkel\surnameend} \&
  \bibinfo{author}{Pierre \surnamestart Simonnet\surnameend}
  (\bibinfo{year}{2009}): \emph{\bibinfo{title}{On Recognizable Tree Languages
  Beyond the {B}orel Hierarchy}}.
\newblock {\sl \bibinfo{journal}{Fundamenta Informaticae}}
  \bibinfo{volume}{95}(\bibinfo{number}{2-3}), pp. \bibinfo{pages}{287--303}.
\newblock \urlprefix\url{http://dx.doi.org/10.3233/FI-2009-151}.

\bibitemdeclare{article}{msou_topol_journal}
\bibitem{msou_topol_journal}
\bibinfo{author}{Szczepan \surnamestart Hummel\surnameend} \&
  \bibinfo{author}{Micha{\l} \surnamestart Skrzypczak\surnameend}
  (\bibinfo{year}{2012}): \emph{\bibinfo{title}{The Topological Complexity of
  {MSO+U} and Related Automata Models}}.
\newblock {\sl \bibinfo{journal}{Fundamenta Informaticae}}
  \bibinfo{volume}{119}(\bibinfo{number}{1}), pp. \bibinfo{pages}{87--111}.

\bibitemdeclare{book}{kechris}
\bibitem{kechris}
\bibinfo{author}{Alexander~S. \surnamestart Kechris\surnameend}
  (\bibinfo{year}{1995}): \emph{\bibinfo{title}{Classical Descriptive Set
  Theory}}.
\newblock {\sl \bibinfo{series}{Graduate Texts in Mathematics}}
  \bibinfo{volume}{156}, \bibinfo{publisher}{Springer-Verlag}.

\bibitemdeclare{book}{moschovakis}
\bibitem{moschovakis}
\bibinfo{author}{Yiannis~N. \surnamestart Moschovakis\surnameend}
  (\bibinfo{year}{2009}): \emph{\bibinfo{title}{Descriptive Set Theory: Second
  Edition}}.
\newblock {\sl \bibinfo{series}{Mathematical Surveys and Monographs}}
  \bibinfo{volume}{155}, \bibinfo{publisher}{American Mathematical Society}.

\bibitemdeclare{inproceedings}{niwinski86}
\bibitem{niwinski86}
\bibinfo{author}{Damian \surnamestart Niwi\'nski\surnameend}
  (\bibinfo{year}{1986}): \emph{\bibinfo{title}{On Fixed-Point Clones (Extended
  Abstract)}}.
\newblock In: {\sl \bibinfo{booktitle}{ICALP}}, pp. \bibinfo{pages}{464--473}.
\newblock \urlprefix\url{http://dx.doi.org/10.1007/3-540-16761-7_96}.

\bibitemdeclare{unpublished}{unambiguous_unpublished}
\bibitem{unambiguous_unpublished}
\bibinfo{author}{Damian \surnamestart Niwi\'nski\surnameend} \&
  \bibinfo{author}{Igor \surnamestart Walukiewicz\surnameend}
  (\bibinfo{year}{1996}): \emph{\bibinfo{title}{Ambiguity problem for automata
  on infinite trees}}.
\newblock \bibinfo{note}{Unpublished note}.

\bibitemdeclare{article}{gap}
\bibitem{gap}
\bibinfo{author}{Damian \surnamestart Niwi\'nski\surnameend} \&
  \bibinfo{author}{Igor \surnamestart Walukiewicz\surnameend}
  (\bibinfo{year}{2003}): \emph{\bibinfo{title}{A gap property of deterministic
  tree languages}}.
\newblock {\sl \bibinfo{journal}{Theor. Comput. Sci.}}
  \bibinfo{volume}{1}(\bibinfo{number}{303}), pp. \bibinfo{pages}{215--231}.
\newblock \urlprefix\url{http://dx.doi.org/10.1016/S0304-3975(02)00452-8}.

\bibitemdeclare{article}{rabin69}
\bibitem{rabin69}
\bibinfo{author}{Michael~O. \surnamestart Rabin\surnameend}
  (\bibinfo{year}{1969}): \emph{\bibinfo{title}{Decidability of Second-Order
  Theories and Automata on Infinite Trees}}.
\newblock {\sl \bibinfo{journal}{Transactions of the AMS}}
  \bibinfo{volume}{141}, pp. \bibinfo{pages}{1--23}.

\bibitemdeclare{article}{rabin70}
\bibitem{rabin70}
\bibinfo{author}{Michael~O. \surnamestart Rabin\surnameend}
  (\bibinfo{year}{1970}): \emph{\bibinfo{title}{Weakly Definable Relations and
  Special Automata}}.
\newblock {\sl \bibinfo{journal}{Mathematical Logic and Foundations of Set
  Theory}}, pp. \bibinfo{pages}{1--23}.

\end{thebibliography}

\end{document}